\documentclass[letterpaper, 10 pt, conference]{ieeeconf}

\IEEEoverridecommandlockouts                              
\usepackage{amssymb,amsfonts,mathtools}
\usepackage{acronym}
\usepackage{amsmath}

\usepackage{amsthm}

\usepackage{multirow}
\usepackage{booktabs}
\usepackage{bm}
\usepackage{tikz}
\usepackage{paralist}
\usepackage{xurl}
\usetikzlibrary{shapes,arrows,fit,backgrounds,automata, positioning}
\usetikzlibrary{shapes.geometric,fit}
\tikzstyle{container} = [draw, rectangle, inner sep=0.3cm]

\usepackage{graphicx} 

\newcommand{\ie}{\textit{i.e.}}

\newcommand{\dist}{\mathcal{D}}

\newcommand{\supp}{\mathsf{Supp}}

\newtheorem{definition}{Definition}
\newtheorem{proposition}{Proposition}

\newtheorem{problem}{Problem}

\newtheorem{remark}{Remark}

\newtheorem{lemma}{Lemma}

\acrodef{pomdp}[POMDP]{partially observable Markov decision process}
\acrodef{mdp}[MDP]{Markov decision process}
\acrodef{asw}[ASW]{Almost-Sure Winning}
\acrodef{cps}[CPS]{Cyber-Physical Systems}
\acrodef{hmm}[HMM]{hidden Markov model}
\acrodef{ltlf}[LTL$_f$]{Linear Temporal Logic over Finite Traces}
\acrodef{ltl}[LTL]{linear temporal logic}
\acrodef{scltl}[co-safe LTL]{syntactically co-safe Linear Temporal Logic}
\acrodef{dfa}[DFA]{deterministic finite automaton}
\acrodef{lstm}[LSTM]{Long-Short Term Memory}

\newcommand{\Always}{\Box \, }
\newcommand{\Eventually}{\Diamond \, }
\newcommand{\Next}{\bigcirc \, }
\newcommand{\until}{\mbox{$\, {\sf U}\,$}}

\newcommand{\probs}{\mathbb{P}}
\newcommand{\Expect}{\mathbb{E}}

 \DeclareMathOperator*{\optmin}{\mathrm{minimize}}

\title{Information-Driven Active Perception for $k$-step Predictive Safety Monitoring}

\author{Sumukha Udupa and Jie Fu
\thanks{S. Udupa, and J. Fu are with the Dept. of Electrical and Computer Engineering, University of Florida, Gainesville, Fl 32611.
{\tt\small \{sudupa, fujie\}@ufl.edu}}}

\begin{document}

\maketitle
\begin{abstract}
    This work studies the synthesis of active perception policies for predictive safety monitoring in partially observable stochastic systems. Operating under strict sensing and communication budgets, the proposed monitor dynamically schedules sensor queries to maximize information gain about the safety of future states. The underlying stochastic dynamics are captured by a labeled \ac{hmm}, with safety requirements defined by a \ac{dfa}. To enable active information acquisition, we introduce minimizing $k$-step Shannon conditional entropy of the safety of future   states as a planning objective, under the constraint of a limited sensor query budget. Using observable operators, we derive an efficient algorithm to compute the $k$-step conditional entropy and     
   analyze key properties of the conditional entropy gradient with respect to policy parameters. We validate the effectiveness of the method for predictive safety monitoring through a dynamic congestion game example. 
\end{abstract}

\section{Introduction}
Runtime verification and predictive monitoring are essential for ensuring the safety of autonomous systems \cite{luckcuck2019formal, yoon2021predictive}. By continuously observing the environment and forecasting future trajectories against formal safety specifications, these monitors anticipate critical failures and enable  timely intervention before critical failures occur. However, the efficacy of traditional predictive monitoring relies heavily on the assumption of high-fidelity state observability. In practice, mobile robots and edge infrastructure operate under strict resource constraints, such as limited wireless bandwidth and sensor blind spots, rendering the environment partially observable. If a monitor acts merely as a passive recipient of intermittent data, its predictive uncertainty diverges rapidly, leading to overly conservative $k$-step safety forecasts.

Motivated by these limitations, we study the synthesis of an active perception policy \cite{bajcsy2018revisiting} for predictive safety monitoring. 
In our proposed method, we define this policy as a mapping from system's observation history to a probability distribution over a discrete set of sensor queries. Rather than passively receiving data, this active perception policy  allocates resource-constrained sensor queries to actively acquire targeted data for   minimizing the $k$-step predictive uncertainty of the safety of future states. To formalize this, 
we model the partially observable stochastic environment as a labeled \ac{hmm} and encode the formal safety specification as a \ac{dfa}.


Our contributions
are:\begin{inparaenum}
    \item We introduce information-driven active perception planning for predictive runtime monitoring in partially observable, stochastic environments, using conditional entropy to quantify the $k$-step predictive uncertainty of a formal safety specification.
    \item We develop a policy gradient approach to synthesize a budget constrained active perception policy for $k$-step predictability and leverage the observable operators \cite{jaeger2000observable} for efficient computation of the prediction probabilities.
    \item We empirically demonstrate 
    the algorithm in a dynamic congestion game, showing that our approach actively minimizes predictive safety errors and 
    closes the performance gap to an oracle with perfect state information.
\end{inparaenum}

\section{Related Work}
Traditional runtime monitoring evaluates execution traces against formal specifications to prevent failures. In \cite{yoon2021predictive}, Bayesian intent inference enables predictive monitors to proactively evaluate formal logic rules and anticipate violations. Complementary work addresses uncertainty in future state predictions using distribution-free techniques such as Conformal Prediction \cite{dietterich2022conformal}, with recent extensions incorporating temporal correlations \cite{tonkens2023scalable} and adaptive bounds for dynamic agents \cite{sheng2024safe}.  These approaches perform well when high-fidelity continuous observations are available. However, under partial observability, the resulting uncertainty bounds tend to widen, leading to overly conservative safety margins. 

Beyond passive monitoring, prediction-informed safe planning embeds forecasting models directly into the control pipeline. To mitigate dynamic model mismatch, recent methods utilize belief-space planning to actively influence agent behavior \cite{pandya2025robots}, integrate Bayesian confidence tracking \cite{fisac2018probabilistically}, adapt trajectories based on prediction variance \cite{kanazawa2019adaptive}, or employ stochastic Control Barrier Functions to enforce safety over predicted distributions \cite{busellato2025uncertainty}. While these methods manage behavioral uncertainty through physical control actions, they typically rely on continuous state tracking and focus on how to safely navigate given an uncertainty distribution. However, they do not address the precursor challenge of \textit{actively scheduling   sensing resources} to reduce observational uncertainty under bandwidth or other resource constraints.

Active perception addresses the precursor problem of constrained data acquisition by scheduling informative sensor queries. The mathematical foundations of this approach are rooted in controlled sensing for \ac{hmm}s \cite{krishnamurthy2002algorithms} and information-theoretic sensor management \cite{hero2008information}. As exact dynamic programming is intractable, modern formulations typically rely on point-based approximations or submodular greedy heuristics to maximize generic information gain \cite{satsangi2020maximizing}. Similarly, task-oriented active sensing via action entropy minimization \cite{greigarn2019task} queries sensors only when they influence the immediate control policy. 
Typically, work on active perception focuses on improving current state estimation or selecting control inputs given available observations. In contrast, our work leverages active perception for predictive monitoring, aiming to resolve uncertainty over a future $k$-step horizon. 

\section{Preliminaries and Problem Formulation}
\noindent\textbf{Notations} Let $X$ be a finite set, and $\dist(X)$ the set of all probability distributions over $X$. For $d\in\dist(X)$, $\supp(d)={x\in X\mid d(x)>0}$ is its support. Let $X^K$ and $X^\ast$ denote sequences of length $K$ and all finite sequences from $X$. $\mathbb{R}$, $\mathbb{N}$, $\mathbb{P}$ denote the reals, naturals, and probability measure. Random variables are capitalized, their realizations lowercase (e.g., $X$ and $x$), and sequences of length $K$ are $X_{0:K}$ and $x_{0:K}$.


We consider active runtime monitoring where an autonomous robot interacts with an active perception agent that infers, from partial observations, whether the robot will enter an unsafe state within $k$ steps. Limited bandwidth restricts sensor access, requiring sensor selection to minimize uncertainty about future unsafe states. The system is modeled as a labeled \ac{hmm} with a controllable emission function.

\begin{definition}[Labeled-\ac{hmm}]
\[
M = \langle S, P, O, \Sigma, \mu_0, \sigma_0, E, \mathcal{AP}, L \rangle
\]
where
\begin{itemize}
    \item $S$ is a finite set of states;
     \item $P:S\rightarrow\dist(S)$ is the probabilistic transition function and $P(s'\mid s)$ is the probability of reaching $s'$ from $s$.
      \item $O$ is the finite set of observations;
      \item $\Sigma$ is a finite set of active perception actions;
      \item $\mu_0$ is the initial state distribution;
      \item $\sigma_0$ is the initial perception action;
    \item $E:S\times\Sigma\rightarrow\dist(O)$ is the probabilistic emission function mapping each state–action pair to a distribution over observations. It is \emph{controllable} since the active perception agent selects the perception action.
   \item $\mathcal{AP}$ is the set of atomic propositions.
   \item $L:S\to 2^{\mathcal{AP}}$ is the labeling function mapping each state to the set of atomic propositions true at that state.
\end{itemize}
\end{definition}

As safety requirements often involve temporal dependencies, we use \ac{ltlf} to represent unsafe conditions.

\begin{definition}[\ac{ltlf}~\cite{de2013linear}]
 An \ac{ltlf} formula over $\mathcal{AP}$ is defined inductively as follows:
\[ \varphi := p \mid \neg\varphi \mid \varphi_1 \land \varphi_2 \mid \varphi_1 \lor \varphi_2 \mid \Next \varphi \mid \varphi_1 \until \varphi_2 \mid \Eventually \varphi \mid \Always \varphi, \]
where $p \in \mathcal{AP}$; $\neg$, $\land$ and $\lor$ are the Boolean operators; and $\Next$, $\until$, $\Eventually$ and $\Always$ denote the temporal modal operators for \texttt{next}, \texttt{until}, \texttt{eventually} and \texttt{always}, respectively.
\end{definition}

For an \ac{ltlf} formula $\varphi$ over $\mathcal{AP}$, the satisfying words are $\mathsf{Words}(\varphi)={w\in(2^{\mathcal{AP}})^*\mid w\models\varphi}$.

\begin{definition}[\ac{dfa}]
\ac{dfa} is a tuple $\mathcal{A}=(Q, 2^{\mathcal{AP}},\delta, q_{init}, F)$ with a finite set of states $Q$, a finite alphabet $2^{\mathcal{AP}}$, a transition function $\delta:Q\times 2^{\mathcal{AP}} \rightarrow Q$, extended recursively for finite words. $q_{init}$ is the initial state, and $F\subseteq Q$ is a set of accepting states.
\end{definition}



Following \cite{de2013linear}, 
the \ac{ltlf} formula is converted to a violation recognizing (Failure) \ac{dfa} accepting $\mathsf{Words}(\varphi)$ to model the coupled system–monitor dynamics.

\begin{definition}[Product \ac{hmm}]
    Given $M=\langle S, P, O, \Sigma, \mu_0, \sigma_0, E, \mathcal{AP}, L \rangle$, and the Failure \ac{dfa} $\mathcal{A}=(Q, 2^{\mathcal{AP}}, \delta, q_0, F_{fail})$, the Product \ac{hmm} is 
    \[
    \mathcal{M} = \langle Z, O, \Sigma, \mathcal{P}, E, \tilde{\mu}_0, F_Z \rangle
    \]
    where
    \begin{itemize}
        \item $Z := S\times Q$ is the product state space, with $z=(s,q)$, where $s\in S$ and $q\in Q$. 
        \item $O$ is a finite set of observations.
        \item $\Sigma$ is the finite set of active perception actions.
        \item $\mathcal{P}: Z \to \dist(Z)$ is the probabilistic transition function such that, given $z=(s, q)$ and $z'=(s',q')$,
        $$ \mathcal{P}(z'\mid z)=P(s'\mid s)\cdot \mathbf{1}(q'=\delta(q,L(s'))),
        $$
        where $\mathbf{1}(\cdot)$ is the indicator function.
        \item $E: Z\times\Sigma \to \dist(O)$ is the emission function. By slight abuse of notation, we extend it to the product space since observations depend only on the physical state $s$ and query $\sigma$, i.e., $E(o\mid (s,q),\sigma)=E(o\mid s,\sigma)$.
        \item $\tilde{\mu}_0$ is the initial product state distribution with $\tilde{\mu}_0(s,q)=\mu_0(s)$ if $q=q_0$, and $0$ otherwise.
        \item $F_Z=\{(s,q)\in Z \mid q\in F_{\text{fail}}\}$ are the failure states.
    \end{itemize}
\end{definition}

A non-stationary observation-based perception policy is a function $\pi:O^\ast \to \dist(\Sigma)$ mapping the observation history $o_{0:t}$ to a distribution $\pi(\cdot\mid o_{0:t})$ over perception actions.

As the active perception agent lacks access to true state, we use an observation-based policy. For a product \ac{hmm} $\mathcal{M}$, a policy $\pi$ induces the stochastic process $\mathcal{M}^\pi:=\{Z_t,\Sigma_t,O_t\}_{t\in\mathbb{N}}$, where $Z_t$, $\Sigma_t$, and $O_t$ denote the product state, perception action, and observation at time $t$.
Motivated by bounded predictability in partially observed systems \cite{cassez2013predictability}, we formalize $k$-step predictive safety as:
\begin{definition}[State-based $k$-step predictability]
    Given the stochastic process $\mathcal{M}^\pi$, finite horizon $T$, and a prediction look-ahead window $k$, let the random variable $W^k_t$ be:
    \[
    W^k_t = \begin{cases}
        1, & \text{if } \exists j\in\{t,\ldots, t+k\} \text{ s.t } Z_j\in F_Z,\\
        0, & \text{otherwise}.
    \end{cases}
    \]
    The \emph{quantitative $k$-step predictability} under $\pi$ is the conditional entropy of $W_t^k$ given the observation history $Y_{0:t}=(O_{0:t}, \Sigma_{0:t})$:
    \vspace{-0.175cm}
    \begin{multline}
    H(W_t^k \mid Y_{0:t}; \pi) =\\
    -\sum_{w_t^k \in \{0,1\}} \sum_{y_{0:t}\in O^t \times \Sigma^t} \probs^\pi (w_t^k,y_{0:t}) \cdot \log \probs^\pi (w_t^k \mid y_{0:t} ).
    \end{multline}
\end{definition}

Lower entropy indicates higher transparency, i.e., the monitor is confident whether a failure is imminent ($W_t^k=1$) or not ($W_t^k=0$). We define the active perception problem with a switching cost $C:\Sigma\times\Sigma\to\mathbb{R}_{\ge 0}$, independent of the robot’s state. If the cost depended on the true state, its real-time observation would act as an implicit emission channel, leaking state information.

\begin{problem}[Predictive active safety monitoring]
Given the Product \ac{hmm} $\mathcal{M}$, finite horizon $K$, and prediction window $k$, synthesize an observation-based policy $\pi \in \Pi$, where $\Pi$ is the policy space, that balances the minimization of the average $k$-step predictability entropy of $W_t^k$ given $Y_{0:t}=(O_{0:t},\Sigma_{0:t})$ and the minimization of total switching costs:
\vspace{-0.2cm}
\begin{align}
\label{eq:optimization-transp-k-step-predictability}
& \optmin_\pi 
\frac{1}{K} \sum_{t=0}^K H(W_t^k|Y_{0:t};  \pi) + \alpha  \sum_{t=0}^{K} \Expect_\pi \left[C(\Sigma_{t-1},\Sigma_{t})\right].
\end{align}
    where $\alpha\ge 0$ is a weighting parameter and the expectation is over the $\pi$-induced stochastic process $M^\pi$.
\end{problem}
\vspace{-0.2cm}



\section{Active perception for $k$-step safety}
Consider a class of parametrized policies $\{\pi_\theta\mid \theta \in \Theta\}$. Let $\mathcal{M}^\theta$ be the stochastic process induced by $\pi_\theta$, $\{Z_t, \Sigma_t, O_t, t\in \mathbb{N}\}$ with $\probs_\theta$ as the corresponding probability measure.
Then, 
given 
$\pi_\theta$, the gradient of $H(W_t^k\mid Y_{0:t}; \theta)$ is 
\begin{equation}
\begin{aligned}
\label{eq:conditional_entropy_gradient_calculation_k_step_predictability_of_failure}
 &\nabla_\theta H(W_t^k\mid Y; \theta) \\
= & - \sum_{y \in O^t\times\Sigma^t} \sum_{w_t^k \in \{0,1\}} \Big[\nabla_\theta \probs_\theta(w_t^k, y) \log \probs_\theta(w_t^k \mid y) \\
&+  \probs_\theta(w_t^k, y)\nabla_\theta \log \probs_\theta(w_t^k\mid y)\Big] \\
 = & - \sum_{y \in O^t\times \Sigma^t} \probs_\theta (y) \sum_{w_t^k \in \{0, 1\}} \Big[ \log \probs_\theta(w_t^k \mid y) \nabla_\theta \probs_\theta(w_t^k\mid y) \\
& + \probs_\theta(w_t^k\mid y) \log \probs_\theta(w_t^k \mid y) \frac{\nabla_\theta \probs_\theta(y)}{\probs_\theta(y)}   + \frac{\nabla_\theta \probs_\theta(w_t^k \mid y)}{\log 2} \Big],
\end{aligned}
\end{equation}
where $\mathcal{Y}_{\theta}^t=\{y\in O^t \times \Sigma^t \mid \probs_\theta(y)>0\}$ is the set of possible observations under the policy $\pi_\theta$ up to a time horizon $t$. 

We compute the gradient using the observable operators.
Given $Z_t, O_t, \Sigma_t$, the reversed transition matrix is $T\in\mathbb{R}^{N\times N}$ with $T_{i,j}=\probs(Z_{t+1}=i\mid Z_t=j)$. For $\sigma\in\Sigma$, the observation matrix is $B^\sigma\in\mathbb{R}^{M\times N}$ with $B^\sigma_{o,j}=E(o\mid j,\sigma)$.

\begin{definition}[Observable operator \cite{jaeger2000observable}]
Given the product \ac{hmm} $\mathcal{M}$, for any observation $o$, perception query $\sigma$, the observable operator $A_{o\mid \sigma}$ is defined as a matrix of size $N \times N$ with its $ij$-th component given as
$$
A_{o\mid \sigma}[i,j]=T_{i,j}B^\sigma_{o, j}
$$
which is the probability of transitioning from state $j$ to state $i$ and at state $j$, observing $o$ given the perception query $\sigma$. This is further expressed in the matrix representation as follows:
$$
A_{o\mid \sigma}=T\text{diag}(B^\sigma_{o, 1},\ldots, B^\sigma_{o,N}).
$$
\end{definition}

We can then express the conditional probability of remaining safe for the
next $k$ steps ($W_t^k=0$) given the history $y$.

\begin{proposition}
    \cite{udupa2025synthesis} The probability of an observation sequence $o_{0:t}$ given 
    an active perception sequence $\sigma_{0:t}$
    is 
    $$
    \probs(o_{0:t}\mid \sigma_{0:t})=\mathbf{1}^\top_N A_{o_t\mid \sigma_t}\ldots A_{o_{0}\mid \sigma_0}\tilde{\mu}_0.
    $$
\end{proposition}

\begin{proposition}
\label{prop:joint_dist_obs_and_arrived_hidden_state}
    \cite{shi2024active} The joint distribution of a sequence of observations, and the arrived hidden state is given as
    $$
    \probs(Z_{t}=i, o_{0:t-1}\mid \sigma_{0:t-1})=\mathbf{1}_i^\top A_{o_{t-1}\mid \sigma_{t-1}}\ldots A_{o_{0}\mid\sigma_0}\tilde{\mu}_0.
    $$
\end{proposition}

\begin{proposition}
    \cite{shi2024active} The probability of a sequence $y=(o_{0:t}, \sigma_{0:t})$ of observations and active sensing queries in $\mathcal{M}^\theta$: 
    $$
    \probs_\theta(o_{0:t}, \sigma_{0:t})=\frac{\probs(o_{0:t}\mid \sigma_{0:t})}{\probs(o_0\mid \sigma_0)}\prod_{i=o}^t \pi_\theta(\sigma_i\mid o_{0:i-1}),
    $$
    where $o_{0:-1}:=\emptyset$ is the initial empty observation.
\end{proposition}

\begin{proposition}
    \cite{shi2024active} Given $y=(o_{0:t}, \sigma_{0:t})$, the gradient of $\log \probs_\theta(y)$ is 
    computed as 
    $$
    \nabla_\theta \log \probs_\theta(y) = \sum_{i=0}^t\nabla_\theta\log\pi_\theta(\sigma_i\mid o_{0:i-1}).
    $$
\end{proposition}

To evaluate safety over $t$ to $t+k$, we have the following.

\begin{proposition}
\label{prop:safety_probability_prop}
    Given $y=(o_{0:t}, \sigma_{0:t})$, the probability $\probs(W_t^k=0\mid y)$ can be written as follows
\begin{multline*} 
\probs(W_t^k=0 \mid y) =\\
\frac{\mathbf{1}^\top (D T)^k D \text{  } \text{diag}(B^{\sigma_t}_{o_t}) \left( A_{o_{t-1}\mid \sigma_{t-1}} \dots A_{o_0\mid \sigma_0} \tilde{\mu}_0 \right)}{\probs(o_{0:t}\mid \sigma_{0:t})}, 
\end{multline*}
where $D\in \mathbb{R}^{N\times N}$ is a diagonal matrix defined as

$$
D(z,z) = \begin{cases}
    1 &\text{if } z\notin F_Z,\\
    0 &\text{otherwise.}
\end{cases}
$$
\end{proposition}
\begin{proof}
    The event $W_t^k=0$ corresponds to the condition that the system remains in the safe state space, \ie, $z_\tau \notin F_Z$ for all $\tau\in \{t,\ldots, t+k\}$.
    We seek the conditional probability: 
    \begin{equation}
    \label{eq:conditional_probability}
    \probs(W_t^k=0\mid o_{0:t}, \sigma_{0:t}) = \frac{\probs(W_t^k=0, o_{0:t}\mid \sigma_{0:t})}{\probs(o_{0:t}\mid \sigma_{0:t})}.
    \end{equation}
    The numerator is the probability of safe trajectories joint with the observation sequence, conditioned on perception actions. 
    Let $\probs(Z_t,o_{0:t}\mid\sigma_{0:t})\in\mathbb{R}^N$ be the joint distribution vector of $Z_t$ and $o_{0:t}$ given $\sigma_{0:t}$, with $i$-th entry $\probs(Z_t=i,o_{0:t}\mid\sigma_{0:t})$.
    \begin{multline}
        \probs(Z_t=i,o_{0:t}\mid \sigma_{0:t})=\\\probs(o_t\mid Z_t=i, \sigma_t)\cdot \probs(Z_t=i, o_{0:t-1}\mid \sigma_{0:t-1}).    
    \end{multline}
    Element-wise multiplication over states $i\in \{1,\ldots, N\}$ is:
    \begin{multline}
        \probs(Z_t,o_{0:t}\mid \sigma_{0:t})=\\\text{diag}(B^{\sigma_t}_{o_t})\cdot \probs(Z_t, o_{0:t-1}\mid \sigma_{0:t-1}).  
    \end{multline}
    By Prop.~\ref{prop:joint_dist_obs_and_arrived_hidden_state}, the column vector is
    \begin{multline}
        \probs(Z_t,o_{0:t}\mid \sigma_{0:t})=\\
        \text{diag}(B^{\sigma_t}_{o_t})\left( A_{o_{t-1}\mid \sigma_{t-1}} \dots A_{o_0\mid \sigma_0} \tilde{\mu}_0 \right).    
    \end{multline}
    To isolate the probability mass strictly in the safe states at time $t$, we left-multiply by the diagonal matrix $D\in \mathbb{R}^{N\times N}$. Matrix $D$ acts as a linear projection operator that zeros out the probability mass associated with any failure state $z\in F_Z$:
    \begin{multline}
        D\probs(Z_t,o_{0:t}\mid \sigma_{0:t})=\\
        D\text{  } \text{diag}(B^{\sigma_t}_{o_t})\left( A_{o_{t-1}\mid \sigma_{t-1}} \dots A_{o_0\mid \sigma_0} \tilde{\mu}_0 \right).    
    \end{multline}
    The resulting vector in $\mathbb{R}^N$ is the joint probability 
    of arriving at a safe state at time $t$ and generating 
    $o_{0:t}$.
    The marginalized joint probability of the arrived state being safe and the observation sequence, conditioned on the query sequence is
    \begin{multline}
        \probs(Z_t\notin F_Z,o_{0:t}\mid \sigma_{0:t})=\\
        \mathbf{1}^\top D\text{  }\text{diag}(B^{\sigma_t}_{o_t})\left( A_{o_{t-1}\mid \sigma_{t-1}} \dots A_{o_0\mid \sigma_0} \tilde{\mu}_0 \right).    
    \end{multline}    

    For $k$-step future horizon, the probability mass is propagated forward using 
    $T$. To ensure the trajectory remains strictly within the safe set at each time step, we iteratively apply the projection matrix $D$ after each transition. The operator $(DT)^k$ tracks the forward-reachable probability mass over exactly $k$ safe transitions. Applying this operator yields the joint probability of safe trajectories and the observation sequence conditioned on query sequence unto time $t+k$:
    \begin{multline}
\probs(W_t^k=0, o_{0:t}\mid \sigma_{0:t})=\probs(\cap_{\tau=t}^{t+k} Z_\tau \notin F_Z, o_{0:t} \mid \sigma_{0:t}) \\= \mathbf{1}^\top (D T)^k D\text{  } \text{diag}(B^{\sigma_t}_{o_t}) \left( A_{o_{t-1}\mid \sigma_{t-1}} \dots A_{o_0\mid \sigma_0} \tilde{\mu}_0 \right)
\end{multline}
Substituting in Eq.~\ref{eq:conditional_probability},
\begin{multline} 
\label{eq:safe_prob_calc}
\probs(W_t^k=0 \mid y) =\\
\frac{\mathbf{1}^\top (D T)^k D\text{  } \text{diag}(B^{\sigma_t}_{o_t}) \left( A_{o_{t-1}\mid \sigma_{t-1}} \dots A_{o_0\mid \sigma_0} \tilde{\mu}_0 \right)}{\probs_\theta(o_{0:t}\mid \sigma_{0:t})}. 
\end{multline}
\end{proof}



\begin{remark}
    Safety violations such as collisions are terminal, making the failure set $F_Z$ absorbing(no transitions lead back to the safe set $Z\setminus F_Z$). Any probability mass entering $F_Z$ over the $k$-step horizon remains there, so $Z_{t+k}\notin F_Z$ implies safety at all prior steps $\tau\in{t,\dots,t+k}$. This simplifies predictive safety computation as the iterative operator $(DT)^kD$ reduces to a single terminal projection $DT^k$, yielding:
    \begin{multline}
\probs(W_t^k=0 \mid y) =\\
\frac{\mathbf{1}^\top D T^k \text{diag}(B^{\sigma_t}_{o_t}) \left( A_{o_{t-1}\mid \sigma_{t-1}} \dots A_{o_0\mid \sigma_0} \tilde{\mu}_0 \right)}{\probs_\theta(o_{0:t}\mid \sigma_{0:t})}.    
\end{multline}
\end{remark}

We can then obtain $\probs(W_t^k=1\mid y)=1-\probs(W_t^k=0\mid y)$.

\begin{lemma}
\label{lemma:gradient_zero}
For a fixed realized history $y\in O^t\times \Sigma^t$, the predictive safety probability $\probs(W_t^k=0\mid y)$ and the conditional entropy $H(W_t^k\mid y)$ are strictly independent of the policy parameters $\theta$. Consequently,
$$
\nabla_\theta \probs(W_t^k = 0 \mid y) = \mathbf{0}, \quad \text{and} \quad \nabla_\theta H(W_t^k \mid y) = \mathbf{0}.
$$
\end{lemma}
\begin{proof}
By Prop.~\ref{prop:safety_probability_prop}, the conditional probability $\probs(W_t^k=0 \mid y)$ is entirely determined by $T$, $D$, $A_{o \mid \sigma}$, and $B^{\sigma_t}_{o_t}$. As these matrices are strictly independent of the policy parameters $\theta$, for a fixed history $y$, the probability acts as a constant with respect to $\theta$, yielding $\nabla_\theta \probs(W_t^k = 0 \mid y) = \mathbf{0}$. As conditional entropy is strictly determined by this probability, 
$\nabla_\theta H(W_t^k \mid y) = \mathbf{0}$.
    
\end{proof}

\begin{lemma}
    The gradient of the conditional entropy with respect to the policy parameters $\theta$ is 
    \begin{multline}
    \nabla_\theta H(W_t^k\mid Y;\theta)=\\\mathbb{E}_{y\sim\pi_\theta}[H(W_t^k\mid Y=y)\nabla_\theta \log{\probs_\theta(y)}].
    \end{multline}
\end{lemma}


\begin{proof}
Using Eq.~\ref{eq:conditional_entropy_gradient_calculation_k_step_predictability_of_failure}, we isolate the components of the gradient. By Lemma~\ref{lemma:gradient_zero}, we know that the conditional predictive safety probability is independent of the policy parameters, meaning $\nabla_\theta \probs(w_t^k\mid y) = \mathbf{0}$. 

Substituting this into the expansion, the first and third terms inside the inner summation strictly vanish. The expression therefore simplifies to:
\begin{equation}
\begin{aligned}
\label{eq:conditional_entropy_gradient_calculation_k_step_predictability_of_failure_modified_shortened}
 &\nabla_\theta H(W_t^k\mid Y) \\
 = & - \sum_{y \in \mathcal{Y}^t_{\theta}} \probs_\theta (y) \sum_{w_t^k \in \{0, 1\}} \Big[ \probs(w_t^k\mid y) \log_2 \probs(w_t^k \mid y) \\&\nabla_\theta \log \probs_\theta(y) \Big] \\
 = & \sum_{y \in \mathcal{Y}^t_{\theta}} \probs_\theta (y) \Big[ H(W_t^k \mid y) \nabla_\theta \log \probs_\theta(y) \Big] \\
 = & \mathbb{E}_{y\sim \pi_\theta} \Big[ H(W_t^k\mid y) \nabla_\theta \log \probs_\theta(y) \Big],
\end{aligned}
\end{equation}
where the expectation is taken over the stochastic trajectories induced by the perception policy $\pi_\theta$.
\end{proof}

To obtain the locally optimal policy parameter $\theta$, we initialize with $\theta_0$ and carryout the following gradient descent at each iteration $\tau\geq 1$,
\begin{multline}
\label{eq:gradient_desent_eq}
    \theta_{\tau+1}=\theta_\tau - \eta [\frac{1}{K}\sum_{t=0}^K\nabla_\theta H(W_t^k\mid Y_{0:t};\theta_\tau)\\+\alpha \sum_{t=0}^K\nabla_\theta\mathbb{E}_\pi[C(\sigma_{t-1}, \sigma_{t})]].
\end{multline}

Note that enumerating each observation sequence is computationally very expensive. Thus, we employ sample approximation to estimate $\nabla_\theta H(W_t^k\mid Y;\theta)$, given $V$ sequences of observations $\{y_1,\ldots,y_V\}$, 
\begin{multline}
\label{eq:approx_entropy}
    \hat{\nabla}_\theta H(W_t^k\mid Y;\theta) =\\ \frac{1}{V}\sum_{v=1}^V H(W_t^k\mid Y=y_v)\nabla_\theta \log \probs_\theta(y_v).
\end{multline}

Similarly, the gradient of the expected active perception query cost at a specific time step $t$ is approximated by:
\begin{multline}
\label{eq:approx_cost}
    \hat{\nabla}_\theta \mathbb{E}_{y \sim \pi_\theta} [C(\sigma_{t-1}, \sigma_t)] =\\ \frac{1}{V}\sum_{v=1}^V C(\sigma_{v, t-1}, \sigma_{v, t}) \nabla_\theta \log \probs_\theta(y_v).
\end{multline}

Substituting Eq.~\ref{eq:approx_entropy} and Eq.~\ref{eq:approx_cost} in Eq.~\ref{eq:gradient_desent_eq}, yields the gradient update for the active perception agent.


\begin{remark}
    Although the formulation assumes binary predictive safety, the information-theoretic objective extends naturally to multi-valued variables. Expanding the codomain of the $k$-step predictability variable to a finite class set enables broader semantic monitoring tasks.
\end{remark}

\section{Experimental Validation}

We implement our proposed predictive monitor in a dynamic congestion game. 
\footnote{Simulations were executed on an Intel Core i7 CPU, 32 GB RAM, and 8GB RTX 3060 GPU (\urlstyle{same}{\url{https://github.com/Sumukha-Udupa/active-perception-for-predictive-safety}}).} 

The environment is modeled as a $45$-node graph representing a pedestrian space (Fig.~\ref{fig:campus_topology_map}). The ego agent, a robot, transports a sensitive reagent from the Space Lab (node $23$) to the Chemistry Lab (node $41$), starting at node $1$ or $44$.

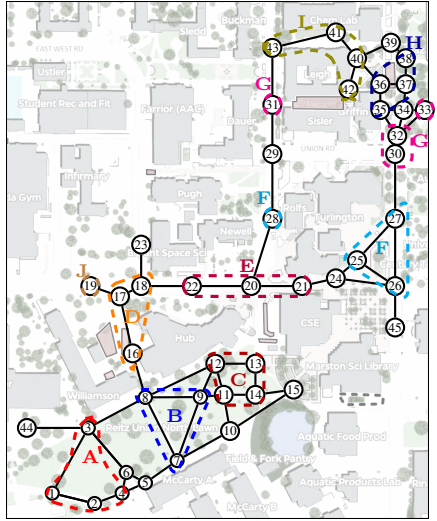
\begin{figure}[htbp]
\vspace{-0.2cm}
\centering
\begin{tikzpicture}
    \node[anchor=south west, inner sep=0] (image) at (0,0) 
    {\includegraphics[width=0.65\columnwidth]{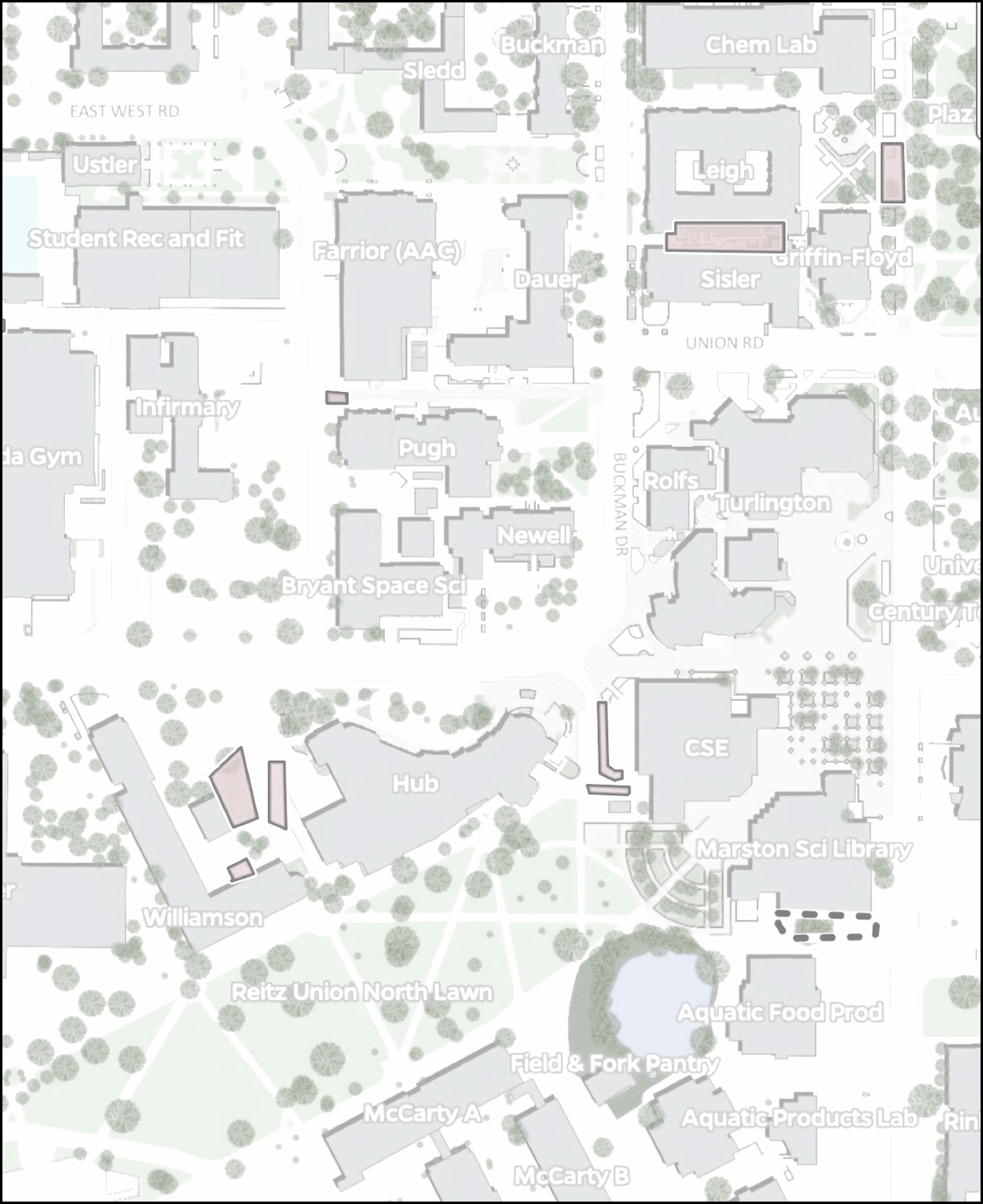}};
    
    \begin{scope}[x={(image.south east)}, y={(image.north west)}]
        
        \tikzstyle{vnode} = [circle, draw=black, fill=white, thick, inner sep=0.25pt, font=\tiny]
        \tikzstyle{startnode} = [circle, draw=black, thick, inner sep=0.25pt, font=\tiny]
        \tikzstyle{goalnode} = [circle, draw=black, thick, inner sep=0.25pt, font=\tiny]
        \tikzstyle{prob} = [midway, fill=white, fill opacity=0.8, text opacity=1, font=\scriptsize, inner sep=0pt, scale=0.9, transform shape]
        

        \node[vnode] (v1)  at (0.11, 0.05) {1};
        \node[vnode] (v2) at (0.21, 0.03) {2};
        \node[vnode] (v3) at (0.195, 0.176) {3};

        \node[vnode] (v4)  at (0.275, 0.05) {4};
        \node[vnode] (v5)  at (0.33, 0.07) {5};
        \node[vnode] (v6)  at (0.285, 0.09) {6};
        \node[vnode] (v7)  at (0.405, 0.113) {7};
        \node[vnode] (v8)  at (0.33, 0.235) {8};
        \node[vnode] (v9)  at (0.46, 0.2356) {9};
        \node[vnode] (v10)  at (0.53, 0.17) {10};
        \node[vnode] (v11)  at (0.515, 0.24) {11};
        \node[vnode] (v12)  at (0.495, 0.3) {12};
        \node[vnode] (v13)  at (0.59, 0.3) {13};
        \node[vnode] (v14)  at (0.59, 0.24) {14};
        \node[vnode] (v15)  at (0.68, 0.25) {15};
        \node[vnode] (v16)  at (0.3, 0.32) {16};
        \node[vnode] (v17)  at (0.27, 0.43) {17};
        \node[vnode] (v18) at (0.32, 0.45) {18};
        \node[vnode] (v19)  at (0.2, 0.45) {19};
        \node[vnode] (v20)  at (0.58, 0.45) {20};
        \node[vnode] (v21)  at (0.7, 0.45) {21};
        \node[vnode] (v22) at (0.44, 0.45) {22};
        \node[vnode] (v23) at (0.32, 0.53) {23};
        \node[vnode] (v24) at (0.78, 0.465) {24};
        \node[vnode] (v25) at (0.83, 0.5) {25};
        \node[vnode] (v26) at (0.92, 0.45) {26};
        \node[vnode] (v27) at (0.92, 0.58) {27};
        \node[vnode] (v28) at (0.63, 0.58) {28};
        \node[vnode] (v29) at (0.63, 0.705) {29};
        \node[vnode] (v30) at (0.92, 0.705) {30};
        \node[vnode] (v31) at (0.63, 0.8) {31};
        \node[vnode] (v32) at (0.925, 0.74) {32};
        \node[vnode] (v33) at (0.99, 0.79) {33};
        \node[vnode] (v34) at (0.94, 0.79) {34};
        \node[vnode] (v35) at (0.885, 0.79) {35};
        \node[vnode] (v36) at (0.885, 0.84) {36};
        \node[vnode] (v37) at (0.945, 0.84) {37};
        \node[vnode] (v38) at (0.945, 0.89) {38};
        \node[vnode] (v39) at (0.91, 0.92) {39};
        \node[vnode] (v40) at (0.83, 0.89) {40};
        \node[vnode] (v41) at (0.78, 0.94) {41};
        \node[vnode] (v42) at (0.81, 0.83) {42};
        \node[vnode] (v43) at (0.63, 0.91) {43};
        \node[vnode] (v44) at (0.05, 0.176) {44};
        \node[vnode] (v45) at (0.92, 0.37) {45};

\draw[-, thick, black] (v1) -- (v2); 
\draw[-, thick, black] (v1) -- (v3);
\draw[-, thick, black] (v2) -- (v4);
\draw[-, thick, black] (v3) -- (v6);
\draw[-, thick, black] (v3) -- (v8);
\draw[-, thick, black] (v4) -- (v5);
\draw[-, thick, black] (v4) -- (v6);
\draw[-, thick, black] (v5) -- (v6); 
\draw[-, thick, black] (v5) -- (v7);
\draw[-, thick, black] (v1) -- (v2); 
\draw[-, thick, black] (v6) -- (v3); 
\draw[-, thick, black] (v7) -- (v8);
\draw[-, thick, black] (v8) -- (v9);
\draw[-, thick, black] (v7) -- (v9);
\draw[-, thick, black] (v8) -- (v16);
\draw[-, thick, black] (v7) -- (v10); 
\draw[-, thick, black] (v10) -- (v11); 
\draw[-, thick, black] (v9) -- (v11); 
\draw[-, thick, black] (v9) -- (v12); 
\draw[-, thick, black] (v8) -- (v12); 
\draw[-, thick, black] (v11) -- (v12); 
\draw[-, thick, black] (v12) -- (v13); 
\draw[-, thick, black] (v13) -- (v14); 
\draw[-, thick, black] (v14) -- (v15);
\draw[-, thick, black] (v10) -- (v15);
\draw[-, thick, black] (v11) -- (v14); 
\draw[-, thick, black] (v16) -- (v17); 
\draw[-, thick, black] (v17) -- (v19); 
\draw[-, thick, black] (v17) -- (v18); 
\draw[-, thick, black] (v18) -- (v23); 
\draw[-, thick, black] (v18) -- (v22); 
\draw[-, thick, black] (v22) -- (v20); 
\draw[-, thick, black] (v20) -- (v21); 
\draw[-, thick, black] (v21) -- (v24);
\draw[-, thick, black] (v20) -- (v28); 
\draw[-, thick, black] (v24) -- (v25); 
\draw[-, thick, black] (v24) -- (v26); 
\draw[-, thick, black] (v26) -- (v45); 
\draw[-, thick, black] (v26) -- (v27); 
\draw[-, thick, black] (v26) -- (v25); 
\draw[-, thick, black] (v25) -- (v27); 
\draw[-, thick, black] (v28) -- (v29);
\draw[-, thick, black] (v27) -- (v30); 
\draw[-, thick, black] (v29) -- (v31); 
\draw[-, thick, black] (v31) -- (v43); 
\draw[-, thick, black] (v41) -- (v43); 
\draw[-, thick, black] (v30) -- (v32); 
\draw[-, thick, black] (v32) -- (v33); 
\draw[-, thick, black] (v32) -- (v34); 
\draw[-, thick, black] (v32) -- (v35);
\draw[-, thick, black] (v34) -- (v35); 
\draw[-, thick, black] (v33) -- (v34); 
\draw[-, thick, black] (v36) -- (v35); 
\draw[-, thick, black] (v34) -- (v37); 
\draw[-, thick, black] (v36) -- (v37); 
\draw[-, thick, black] (v36) -- (v40);
\draw[-, thick, black] (v37) -- (v38); 
\draw[-, thick, black] (v38) -- (v39);
\draw[-, thick, black] (v40) -- (v39); 
\draw[-, thick, black] (v42) -- (v40);
\draw[-, thick, black] (v40) -- (v41);
\draw[-, thick, black] (v44) -- (v3); 
        
        \draw[draw=red, very thick, dashed, rounded corners=6pt] 
            (0.08, 0.04) -- (0.21, 0.01) -- (0.30, 0.04) -- (0.22, 0.12) -- (0.22, 0.20) -- (0.17, 0.19) -- (0.15, 0.10) -- cycle;
        \node[font=\scriptsize, red] at (0.2, 0.12) {$\mathbf{A}$};
        
        \draw[draw=blue, very thick, dashed, rounded corners=4pt] 
            (0.40, 0.08) -- (0.49, 0.25) -- (0.30, 0.25) -- cycle;
        \node[font=\scriptsize, blue] at (0.4, 0.2) {$\mathbf{B}$};
        
        \draw[draw=red!70!black, very thick, dashed, rounded corners=4pt] 
            (0.50, 0.22) -- (0.61, 0.22) -- (0.61, 0.32) -- (0.47, 0.32) -- (0.50, 0.26) -- cycle;
        \node[font=\scriptsize, red!70!black] at (0.55, 0.27) {$\mathbf{C}$};
        
        \draw[draw=orange, very thick, dashed, rounded corners=4pt] 
            (0.24, 0.45) -- (0.35, 0.48) -- (0.31, 0.29) -- (0.28, 0.29) -- cycle;
        \node[font=\scriptsize, orange] at (0.3, 0.39) {$\mathbf{D}$};
        
        \draw[draw=purple, very thick, dashed, rounded corners=4pt] 
            (0.42, 0.43) rectangle (0.72, 0.47);
        \node[font=\scriptsize, purple] at (0.57, 0.49) {$\mathbf{E}$};
        
        \draw[draw=cyan, very thick, dashed, rounded corners=4pt] 
            (0.79, 0.50) -- (0.95, 0.42) -- (0.95, 0.62) -- cycle;
        \draw[draw=cyan, very thick, dashed] 
            (0.63, 0.58) circle (0.02);
        \node[font=\scriptsize, cyan!80!black] at (0.89, 0.525) {$\mathbf{F}$};
        \node[font=\scriptsize, cyan!80!black] at (0.61, 0.62) {$\mathbf{F}$};
        
        \draw[draw=magenta, very thick, dashed, rounded corners=4pt] 
            (0.89, 0.68) rectangle (0.96, 0.76);
        \draw[draw=magenta, very thick, dashed] 
            (0.63, 0.80) circle (0.02);
        \draw[draw=magenta, very thick, dashed] 
            (0.99, 0.79) circle (0.02);
        \node[font=\scriptsize, magenta] at (0.61, 0.84) {$\mathbf{G}$};
        \node[font=\scriptsize, magenta] at (0.98, 0.73) {$\mathbf{G}$};
        
        \draw[draw=blue!60!black, very thick, dashed, rounded corners=4pt] 
            (0.865, 0.77) -- (0.865, 0.86) -- (0.965, 0.91) -- (0.965, 0.82) -- cycle;
        \node[font=\scriptsize, blue!60!black] at (0.965, 0.92) {$\mathbf{H}$};
        
        \draw[draw=olive, very thick, dashed, rounded corners=6pt] 
            (0.61, 0.89) -- (0.61, 0.93) -- (0.78, 0.96) -- (0.85, 0.91) -- (0.83, 0.81) -- (0.79, 0.81) -- (0.79, 0.87) -- (0.75, 0.91) -- cycle;
        \node[font=\scriptsize, olive] at (0.70, 0.96) {$\mathbf{I}$};
        
        \draw[draw=brown, very thick, dashed] 
            (0.20, 0.45) circle (0.02);
        \node[font=\scriptsize, brown] at (0.18, 0.48) {$\mathbf{J}$};

    \end{scope}
\end{tikzpicture}
\vspace{-0.275cm}
\caption{Environment topological graph with sensor coverages.}
\label{fig:campus_topology_map}
\vspace{-0.2cm}
\end{figure}



\noindent \textbf{Environment Dynamics} 
The ego agent follows the stochastic policy given by the transition probabilities in Fig.~\ref{fig:campus_map_fig}. Two independent dynamic obstacles perform stochastic random walks on the graph. The first patrols nodes $3$–$22$ (${3,5,6,7}$ as initial), and the second $29$–$40$ (${29,30}$ as initial). As the zones are disjoint, the ego agent encounters at most one obstacle, so the monitor tracks only the locally relevant obstacle. The second agent activates after the ego exits $3$-$22$. The joint state is $(s_e,s_t)$ for the ego and traffic agents.

\begin{figure}[htbp]
\vspace{-0.3cm}
    \centering
    \resizebox{0.75\columnwidth}{!}{
        \begin{tikzpicture}[x=0.82cm, y=0.8cm, >=stealth]
    
    \tikzstyle{state} = [circle, draw=black, thin, fill=white, inner sep=0.5pt, font=\tiny, minimum size=0.35cm]
    \tikzstyle{start} = [circle, draw=black, thick, fill=white, inner sep=0.5pt, font=\tiny, minimum size=0.35cm]
    \tikzstyle{goal}  = [circle, draw=black, thick, fill=white, inner sep=0.5pt, font=\tiny, minimum size=0.35cm]
    \tikzstyle{prob}  = [midway, fill=white, fill opacity=0.8, text opacity=1, font=\tiny, inner sep=0pt, scale=0.5, transform shape]
    
    
    \node[start] (v1)  at (0, 0) {1};
    \node[start] (v44) at (0, -1) {44};
    \node[state] (v17) at (0, -2) {17};
    \node[state] (v19) at (0, -3) {19};
    
    \node[state] (v2)  at (1, 0) {2};
    \node[state] (v3)  at (1, -1) {3};
    \node[state] (v16) at (1, -2) {16};
    \node[state] (v18) at (1, -3) {18};
    \node[state] (v23) at (1, -4) {23}; 
    
    \node[state] (v4)  at (2, 0) {4};
    \node[state] (v6)  at (2, -1) {6};
    \node[state] (v8)  at (2, -2) {8};
    \node[state] (v22) at (2, -3) {22};
    
    \node[state] (v5)  at (3, 0) {5};
    \node[state] (v7)  at (3, -1) {7};
    \node[state] (v9)  at (3, -2) {9};
    \node[state] (v12) at (3, -3) {12};
    \node[state] (v20) at (3, -4) {20};
    
    \node[state] (v28) at (4, 0) {28};
    \node[state] (v10) at (4, -1) {10};
    \node[state] (v11) at (4, -2) {11};
    \node[state] (v13) at (4.5, -3) {13}; 
    \node[state] (v21) at (4, -4) {21};
    
    \node[state] (v29) at (5, 0) {29};
    \node[state] (v15) at (5, -1) {15};
    \node[state] (v14) at (5, -2) {14};
    \node[state] (v24) at (4.8, -3.4) {24}; 
    
    \node[state] (v25) at (5.5, -2.5) {25}; 
    \node[state] (v26) at (5.5, -3.85) {26}; 
    \node[state] (v30) at (6, 0) {30};
    \node[state] (v27) at (6, -3) {27}; 
    \node[state] (v33) at (6.3, -1) {33}; 
    \node[state] (v31) at (6.5, -2.5) {31}; 
    
    \node[state] (v32) at (7, 0) {32};
    \node[state] (v34) at (7.4, -1) {34};
    
    \node[state] (v35) at (8, 0) {35};
    \node[state] (v36) at (8, -1) {36};
    \node[state] (v37) at (8, -2) {37};
    \node[state] (v38) at (8, -3) {38};
    \node[state] (v43) at (8.5, -4) {43}; 
    
    \node[state] (v41)  at (9, 0) {41};
    \node[state] (v40) at (9, -1) {40};
    \node[state] (v42) at (8.4, -1.6) {42}; 
    \node[state] (v39) at (9, -3) {39};

    \begin{scope}[every path/.style={->, thin, black, shorten >=0.5pt, shorten <=0.5pt}]
        
        \draw (v1) -- node[prob] {$0.5$} (v2);
        \draw (v1) -- node[prob] {$0.5$} (v3);
        \draw (v44) -- node[prob] {$1.0$} (v3);
        \draw (v2) -- node[prob] {$1.0$} (v4);
        \draw (v3) to[bend left=10] node[prob] {$0.2$} (v6);
        \draw (v3) -- node[prob] {$0.8$} (v8);
        
        \draw (v4) -- node[prob] {$0.7$} (v5);
        \draw (v4) -- node[prob] {$0.3$} (v6);
        \draw (v5) to[bend right=10] node[prob] {$0.5$} (v6);
        \draw (v5) -- node[prob] {$0.5$} (v7);
        \draw (v6) to[bend left=10] node[prob] {$0.5$} (v3);
        \draw (v6) to[bend right=10] node[prob] {$0.5$} (v5);
        \draw (v7) to[bend right=15] node[prob] {$0.7$} (v8);
        \draw (v7) -- node[prob] {$0.15$} (v9);
        \draw (v7) -- node[prob] {$0.15$} (v10);
        
        \draw (v8) to[bend left=10] node[prob] {$0.1$} (v9);
        \draw (v8) -- node[prob] {$0.2$} (v12);
        \draw (v8) -- node[prob] {$0.7$} (v16);
        \draw (v9) to[bend left=10] node[prob] {$0.45$} (v8);
        \draw (v9) -- node[prob] {$0.45$} (v12);
        \draw (v9) to[bend left=10] node[prob] {$0.1$} (v11);
        \draw (v10) -- node[prob] {$0.7$} (v11);
        \draw (v10) -- node[prob] {$0.3$} (v15);
        \draw (v11) -- node[prob] {$0.5$} (v12);
        \draw (v11) to[bend left=10] node[prob] {$0.5$} (v9);
        
        \draw (v12) -- node[prob] {$0.75$} (v8); 
        \draw (v12) to[bend left=10] node[prob] {$0.25$} (v13);
        \draw (v13) to[bend left=10] node[prob] {$0.5$} (v12);
        \draw (v13) to[bend right=15] node[prob] {$0.1$} (v14);
        \draw (v13) to[bend left=15] node[prob] {$0.4$} (v20);
        \draw (v14) -- node[prob] {$0.7$} (v11);
        \draw (v14) to[bend left=10] node[prob] {$0.15$} (v15);
        \draw (v14) to[bend right=10] node[prob] {$0.15$} (v13);
        \draw (v15) to[loop above] node[prob, pos=0.2] {$0.4$} (v15); 
        \draw (v15) to[bend left=10] node[prob] {$0.4$} (v14);
        \draw (v15) -- node[prob] {$0.2$} (v10);
        
        \draw (v16) -- node[prob] {$1.0$} (v17);
        \draw (v17) to[bend right=20] node[prob] {$0.25$} (v19);
        \draw (v17) to[bend left=10] node[prob] {$0.75$} (v18);
        \draw (v18) to[bend left=10] node[prob] {$0.05$} (v17);
        \draw (v18) -- node[prob] {$0.45$} (v23);
        \draw (v18) to[bend left=10] node[prob] {$0.5$} (v22);
        \draw (v19) to[bend right=20] node[prob] {$0.5$} (v17);
        \draw (v19) to[loop below] node[prob] {$0.5$} (v19); 
        
        \draw (v20) to[bend right=40] node[prob, right] {$0.55$} (v28); 
        \draw (v20) to[bend right=15] node[prob] {$0.2$} (v22);
        \draw (v20) -- node[prob] {$0.2$} (v21);
        \draw (v20) to[bend left=10] node[prob] {$0.05$} (v13);
        \draw (v21) -- node[prob] {$1.0$} (v24);
        \draw (v22) to[bend left=15] node[prob] {$0.65$} (v18);
        \draw (v22) to[bend right=20] node[prob] {$0.35$} (v20);
        \draw (v23) to[bend left=20] node[prob] {$0.9$} (v18); 
        \draw (v23) to[loop left] node[prob] {$0.1$} (v23);
        
        \draw (v24) -- node[prob] {$0.5$} (v25);
        \draw (v24) -- node[prob] {$0.5$} (v26);
        \draw (v25) -- node[prob] {$1.0$} (v27);
        \draw (v26) -- node[prob] {$1.0$} (v27);
        \draw (v27) -- node[prob, pos=0.25, left] {$1.0$} (v30);
        \draw (v28) -- node[prob] {$1.0$} (v29);
        \draw (v29) -- node[prob] {$0.7$} (v30);
        \draw (v29) to[bend left=10] node[prob] {$0.3$} (v31);
        \draw (v30) -- node[prob] {$1.0$} (v32);
        \draw (v31) to[bend left=10] node[prob] {$0.1$} (v29);
        \draw (v31) -- node[prob, pos=0.75] {$0.9$} (v43);
        
        \draw (v32) to[bend left=10] node[prob] {$0.05$} (v33);
        \draw (v32) -- node[prob] {$0.5$} (v34);
        \draw (v32) -- node[prob] {$0.45$} (v35);
        \draw (v33) to[bend left=10] node[prob] {$0.1$} (v32);
        \draw (v33) -- node[prob] {$0.85$} (v34);
        \draw (v33) to[loop above] node[prob] {$0.05$} (v33);
        \draw (v34) -- node[prob] {$0.5$} (v37);
        \draw (v34) -- node[prob] {$0.5$} (v35);
        \draw (v35) -- node[prob] {$1.0$} (v36);
        \draw (v36) -- node[prob] {$1.0$} (v40);
        \draw (v37) to[bend left=10] node[prob] {$0.9$} (v36);
        \draw (v37) -- node[prob] {$0.1$} (v38);
        \draw (v38) -- node[prob] {$1.0$} (v39);
        \draw (v39) -- node[prob] {$1.0$} (v40);
        
        \draw (v40) -- node[prob] {$0.8$} (v41);
        \draw (v40) to[bend left=10] node[prob] {$0.2$} (v42);
        \draw (v41) to[loop left] node[prob, pos=0.2] {$1.0$} (v41); 
        \draw (v42) to[bend left=10] node[prob] {$1.0$} (v40);
        \draw (v43) to[bend right=55] node[prob, right] {$1.0$} (v41);
        
    \end{scope}

\end{tikzpicture} 
    }
    \vspace{-0.3cm}
    \caption{State transition graph representing the ego agent's goal policy.}
    \label{fig:campus_map_fig}
    \vspace{-0.3cm}
\end{figure}
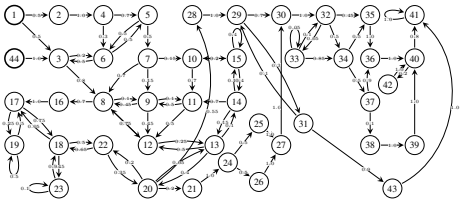

We require ego to visit node $23$ before $41$. Collisions or reaching  $41$ first yield an absorbing failure state. Formally, $\varphi=(\neg \text{node 41}\until \text{node 23})\land \Always \neg \text{crash}$ with \ac{dfa} as in Fig.~\ref{fig:campus_example_dfa}.

\begin{figure}[h]
\vspace{-0.3cm}
    \centering  


\begin{tikzpicture}[shorten >=1pt, >=stealth, auto, semithick, every state/.style={minimum size=0.2cm, inner sep=0pt}]
   
   \node[state, initial, initial text=start] (q0) at (0, 0) {$q_0$};
   \node[state] (q1) at (4, 0) {$q_1$};
   \node[state, accepting, fill=red!10] (qfail) at (2, -1) {$q_{\text{fail}}$};

   \path[->]
    (q0) edge [loop above] node {$\neg $23$ \land \neg $41$\land \neg crash$} (q0)
         edge node [above] {$23 \land \neg 41\land \neg crash$} (q1)
         edge node [below left] {$crash\lor (\neg 23\land 41)$} (qfail)
    (q1) edge [loop above] node {$\neg 41\land \neg crash$} (q1)
         edge node [below right] {$crash$} (qfail)
    (qfail) edge [loop right] node {$1$} (qfail);
\end{tikzpicture} 
   \vspace{-1.1cm}
    \caption{Failure \ac{dfa} for the dynamic congestion game example.}
    \vspace{-0.3cm}
\label{fig:campus_example_dfa}
\end{figure}
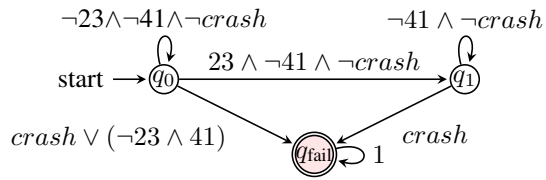



A safety monitor in the environment enforces $k$-step predictive safety. It lacks access to the joint state, relying instead on active perception via $10$ localized proximity boolean sensors with partial node coverage: $\mathbf{A}:\{1,2,3,4\}$, $\mathbf{B}:\{7,8,9\}$, $\mathbf{C}:\{11,12,13,14\}$, $\mathbf{D}:\{16,17,18\}$, $\mathbf{E}:\{20,21,22\}$, $\mathbf{F}:\{25,26,27,28\}$, $\mathbf{G}:\{30,31,32,33\}$, $\mathbf{H}:\{35,36,37,38\}$, $\mathbf{I}:\{40,41,42,43\}$, and $\mathbf{J}:\{19\}$.

The monitor uses horizon $K=30$. At each step $t$, it queries one sensor or takes ``no query''. Maintaining an active sensor costs $5$, switching $10$, and idling $0$. Queries return $\{\text{S}_{\text{ego}},\text{S}_{\text{traffic}},\text{S}_{\text{both}}\}$ or null with a false negative rate of $0.15$.



We parametrize the active perception policy $\pi_\theta$ using a \ac{lstm} to encode the observation history into an implicit information state. At each time step $t$, the hidden state updates as $h_t = \text{LSTM}(o_t, \sigma_t, h_{t-1};\theta_{rnn})$, which is then mapped to a probability distribution over the action space $\Sigma$ via softmax activation.

\noindent \textbf{Results and Discussion}
Predictive performance is benchmarked against two baselines:
\begin{enumerate}
\item \textbf{Uniform Random:} Uniform sensor queries.
\item \textbf{Oracle:} Perfect, cost-free observability of the true product state $z_t$ (performance upper bound).
\end{enumerate}

To evaluate the monitor, we use the predictive safety probability $\probs(W_t^k=0\mid y_t)$ from Eq.~\ref{eq:safe_prob_calc}, which measures the it’s situational awareness and early-warning accuracy under limited sensor cost. Monte-Carlo simulations over sampled trajectories estimate the posterior prediction accuracy. Performance is then quantified via the mean Brier score and gap closure, the fraction of the Random–Oracle baseline gap closed by the learned policy:
$$\text{Gap Closure} = \frac{B_{\text{random}} - B_{\text{lstm}}}{B_{\text{random}} - B_{\text{oracle}}},$$ 
where $B_x$ is the Brier score under policy $x$. 
\begin{table*}[htbp] 
\centering
\caption{Predictive Safety Performance (Mean Brier Score $\pm$ 95\% CI)}
\vspace{-0.3cm}
\label{tab:predictive_error_comparison}
\footnotesize 
\begin{tabular}{@{}lccccccc@{}}
\toprule
\multirow{2}{*}{\textbf{Horizon ($k$)}} & \multicolumn{2}{c}{\textbf{Uniform Random}} & \multicolumn{2}{c}{\textbf{Trained Policy}} & \multirow{2}{*}{\textbf{Oracle}} & \multirow{2}{*}{\textbf{\% Imprv.}} & \multirow{2}{*}{\textbf{Gap Closure}} \\ \cmidrule(lr){2-3} \cmidrule(lr){4-5}
 & \textbf{Brier Score} & \textbf{Cost} & \textbf{Brier Score} & \textbf{Cost} &  &  &  \\ \midrule
$k=1$ & $0.1791 \pm 0.0075$ & $227.07 \pm 4.00$ & $\mathbf{0.0564 \pm 0.0029}$ & $\mathbf{200.65 \pm 2.24}$ & $0.0149 \pm 0.0006$ & $68.53\%$ & $74.72\%$ \\
$k=3$ & $0.1880 \pm 0.0066$ & $224.07 \pm 4.07$ & $\mathbf{0.0799 \pm 0.0035}$ & $\mathbf{198.11 \pm 2.00}$ & $0.0386 \pm 0.0016$ & $57.50\%$ & $72.33\%$ \\
$k=5$ & $0.1931 \pm 0.0058$ & $230.31 \pm 3.80$ & $\mathbf{0.0939 \pm 0.0040}$ & $\mathbf{195.84 \pm 1.87}$ & $0.0576 \pm 0.0025$ & $51.38\%$ & $73.21\%$ \\
$k=10$ & $0.2007 \pm 0.0055$ & $224.57 \pm 4.14$ & $\mathbf{0.1255 \pm 0.0049}$ & $\mathbf{194.55 \pm 1.78}$ & $0.0921 \pm 0.0040$ & $37.47\%$ & $69.24\%$ \\
$k=15$ & $0.2012 \pm 0.0069$ & $230.18 \pm 3.81$ & $\mathbf{0.1395 \pm 0.0059}$ & $\mathbf{194.85 \pm 1.91}$ & $0.1124 \pm 0.0051$ & $30.66\%$ & $69.48\%$ \\
\bottomrule
\end{tabular}
\vspace{-0.6cm}
\end{table*}

Table~\ref{tab:predictive_error_comparison} reports predictive performance of the baselines and learned policy for $k \in \{1,3,5,10,15\}$. Fig.~\ref{fig:campus_example_accuracy_plot} shows posterior prediction over $1000$ sample trajectories. Fig.~\ref{fig:campus_example_convergence} shows training convergence for $k=3$ and perception penalty $\alpha=0.04$. The learning dynamics reflect a multi-objective trade-off: early training prioritizes uncertainty reduction, rapidly lowering the average predictive conditional entropy.
Later, the cost penalty $\alpha$ dominates, yielding resource savings with only a marginal increase in predictive uncertainty.  

\begin{figure}
    \centering
    \includegraphics[width=0.65\linewidth]{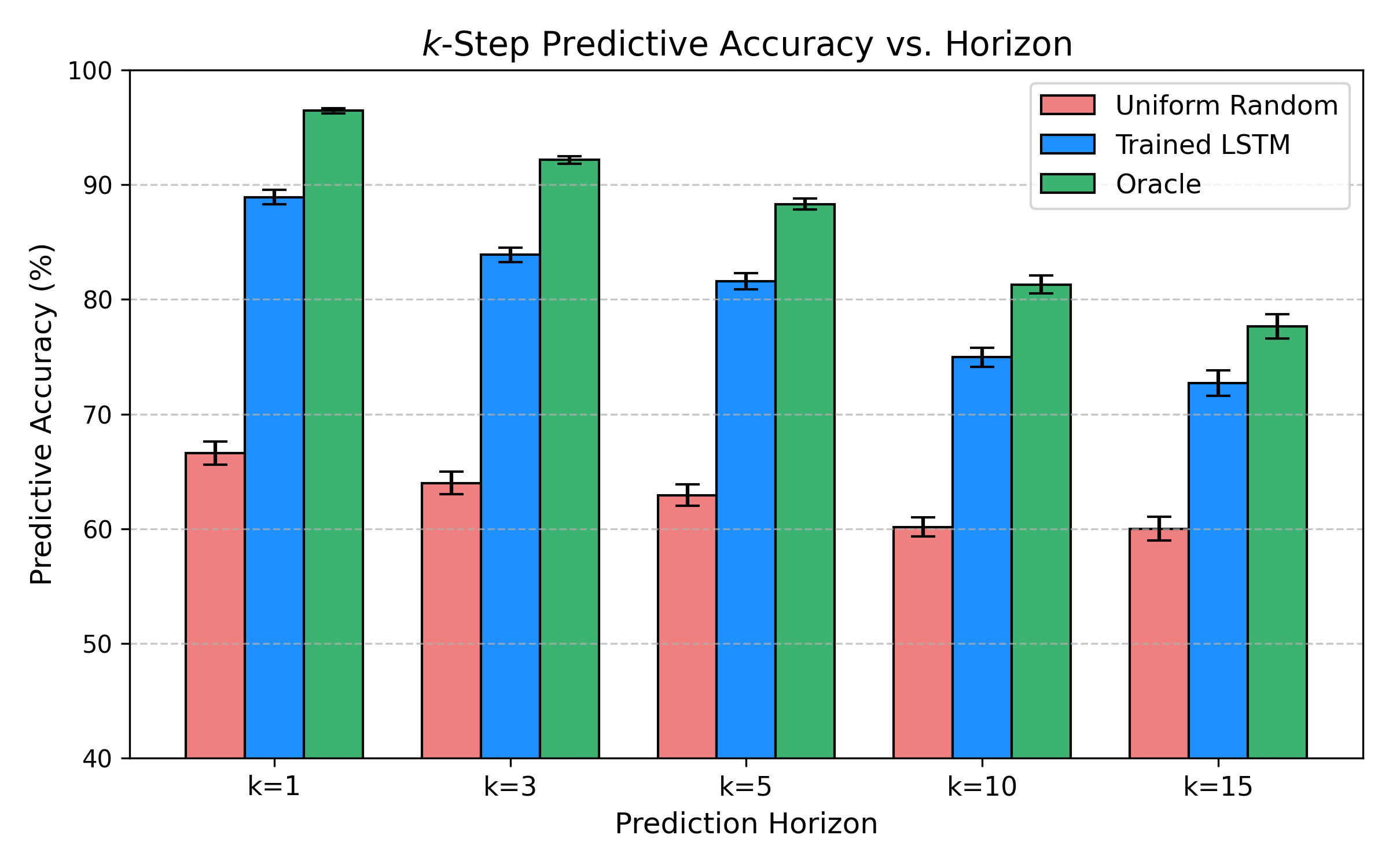}
    \vspace{-0.4cm}
    \caption{Comparison of  $k$-step prediction accuracy with posterior sampling.}
    \label{fig:campus_example_accuracy_plot}
    \vspace{-0.5cm}
\end{figure}

\begin{figure}
\vspace{0cm}
    \centering
    \includegraphics[width=0.85\linewidth]{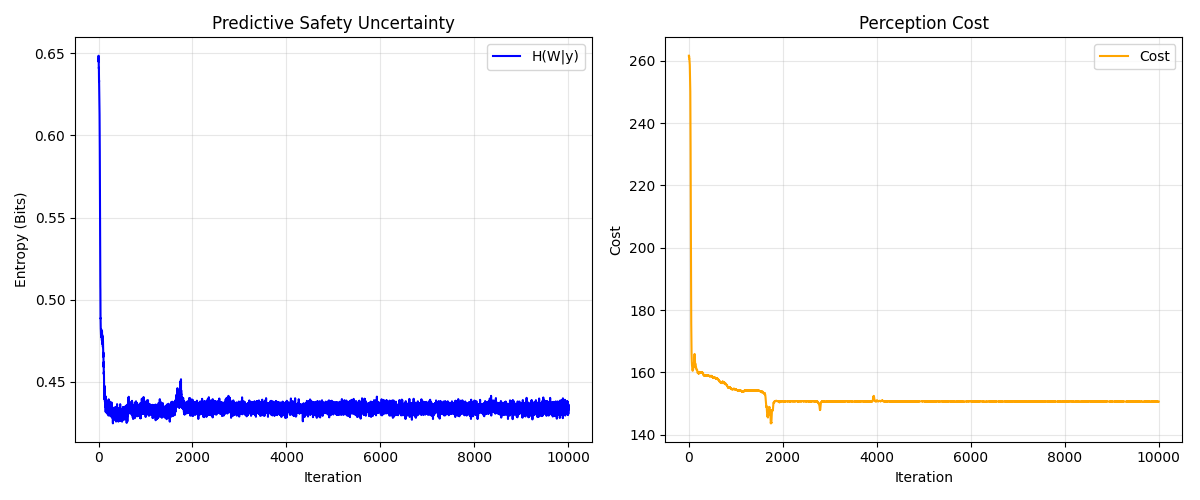}
    \vspace{-0.4cm}
    \caption{Convergence of the policy gradient method for $k=3$, $\alpha=0.04$.}
    \label{fig:campus_example_convergence}
    \vspace{-0.4cm}
\end{figure}

Table~\ref{tab:predictive_error_comparison} shows the learned policy outperforms the random baseline and approaches the Oracle upper bound. Brier scores increase with prediction horizon due to stochastic uncertainty, yet the policy maintains a $30.66\%$–$68.53\%$ improvement over the random baseline.

We evaluate the trade-off between predictive safety and perception cost by varying $\alpha$ from $0.0$ to $0.07$. Fig.~\ref{fig:campus_example_alpha_scatter} shows the Brier scores for $k=3$.
\begin{figure}
    \centering
    \includegraphics[width=0.65\linewidth]{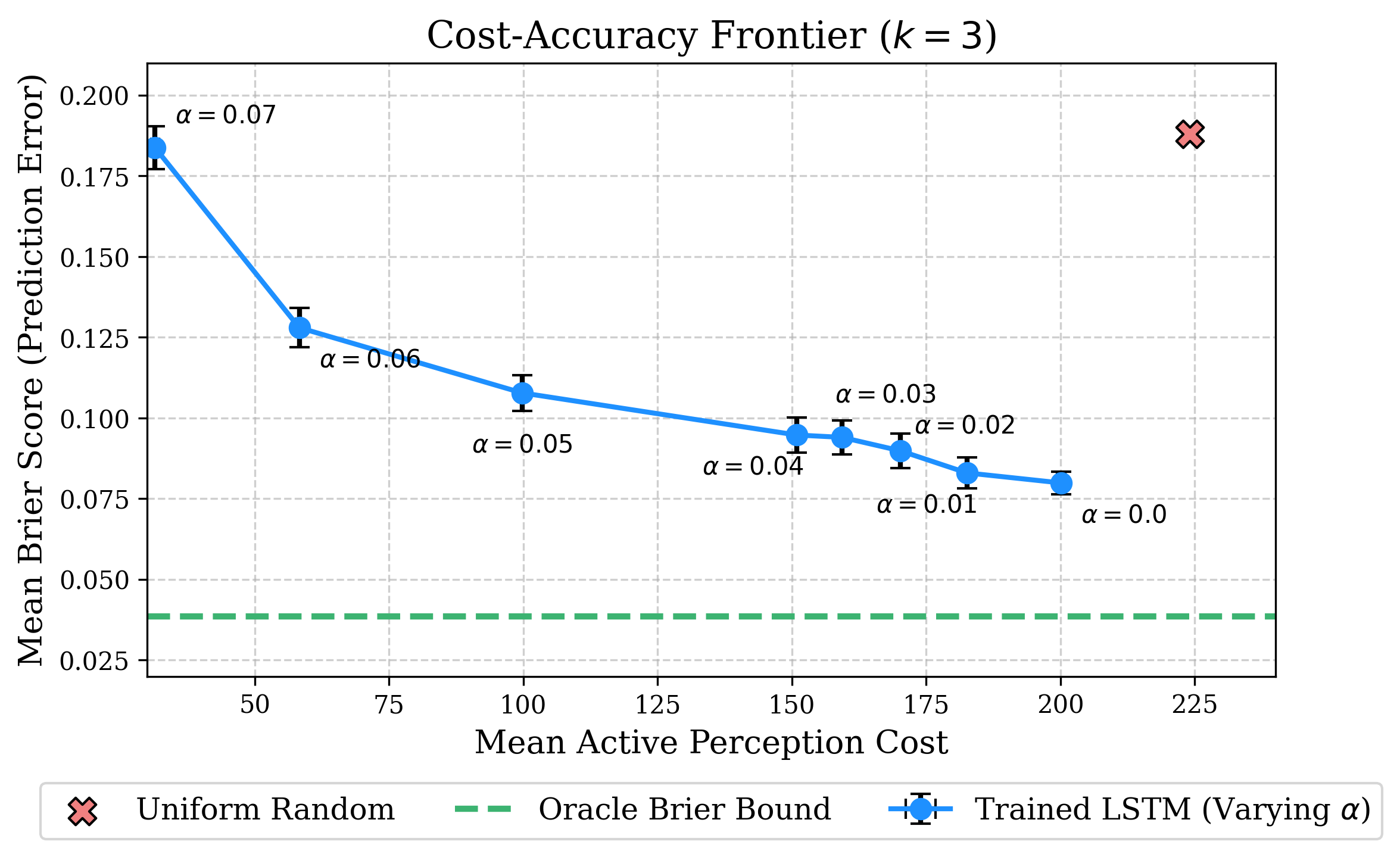}
    \vspace{-0.35cm}
    \caption{Comparison of Brier scores for varying $\alpha$.}
    \vspace{-0.8cm}
    \label{fig:campus_example_alpha_scatter}
\end{figure}
As $\alpha$ increases, the learned policy suppresses unnecessary actions, reducing mean cost. It still outperforms the baseline at high penalties (e.g., $\alpha=0.06$), with about $73\%$ lower cost than random policy. Fig.~\ref{fig:campus_example_alpha_scatter} also shows slightly higher Brier scores for $\alpha \in [0.01,0.04]$, with similar performance at costs $150.9$ vs.\ $182.58$, implying a small increase in prediction uncertainty.

\section{Conclusion}

This letter presented an information-theoretic active perception planning method for $k$-step predictive safety monitoring under partial observability. The proposed approach enables active monitors to selectively acquire informative observations in safety-critical cyber-physical systems with resource-constrained sensing and communication capabilities. By explicitly targeting uncertainty that affects future safety predictions, the method complements existing runtime monitoring and predictive verification techniques. Moreover, it can be integrated with conformal prediction to manage uncertainty bounds or combined with fault-prognosis-driven control strategies to enable timely intervention. Future work will extend the framework to multi-modal sensing systems with continuous observations and investigate scalable implementations for large-scale autonomous systems.


\bibliography{refs}
\bibliographystyle{ieeetr}
\end{document}